\documentclass[11pt]{llncs}

\usepackage{microtype}

\usepackage{wrapfig}
\usepackage{microtype}
\usepackage{graphicx}
\usepackage{graphics}
\usepackage{epsfig}
\usepackage{epstopdf}
\usepackage{amsmath}
\usepackage{latexsym}
\usepackage{wrapfig}
\usepackage{multirow}
\usepackage{amsfonts}
\usepackage{cite}
\usepackage{amssymb}
\usepackage{caption}
\usepackage{algorithm}
\usepackage{algorithmic}
\usepackage{enumerate}
\usepackage{diagbox}
\usepackage{subfig}
\usepackage{wrapfig}
\usepackage{color}

\usepackage[a4paper,margin=2.5cm,footskip=0.25in]{geometry}

\makeatletter
\newcommand{\rmnum}[1]{\romannumeral #1}
\newcommand{\Rmnum}[1]{\expandafter\@slowromancap\romannumeral #1@}
\makeatother

\spnewtheorem{claim}{Claim}{\bfseries}{\rmfamily}
\begin{document}

\title{Greedy Strategy Works for $k$-Center Clustering with Outliers and Coreset Construction}
\author{Hu Ding   \hspace{0.25in} Haikuo Yu\hspace{0.25in} Zixiu Wang }
\institute{
School of Computer Science and Technology\\
University of Science and Technology of China\\
 \email{\tt huding@ustc.edu.cn, yhk7786@mail.ustc.edu.cn, wzx2014@mail.ustc.edu.cn}\\
}
\maketitle

\thispagestyle{empty}


\begin{abstract}
We study the problem of $k$-center clustering with outliers in arbitrary metrics and Euclidean space. Though a number of methods have been developed in the past decades, it is still quite challenging to design quality guaranteed algorithm with low complexity for this problem. 
Our idea is inspired by the greedy method, Gonzalez's algorithm, for solving the problem of ordinary $k$-center clustering. Based on some novel observations, we show that this greedy strategy actually can handle $k$-center clustering with outliers efficiently, in terms of clustering quality and time complexity.  We further show that the greedy approach yields small coreset for the problem in doubling metrics, so as to reduce the time complexity  significantly. Our algorithms are easy to implement in practice. We test our method on both synthetic and real datasets. The experimental results suggest that our algorithms can achieve near optimal solutions and yield lower running times comparing with existing methods. 
 \end{abstract}

%
%
%
  
  \newpage

\pagestyle{plain}
\pagenumbering{arabic}
\setcounter{page}{1}

\section{Introduction}
\label{sec-intro}
\vspace{-0.1in}
{\em Clustering} is one of the most fundamental problems in data analysis~\cite{jain2010data}. Given a set of elements, the goal of clustering is to partition the set into several groups based on their similarities or dissimilarities.
Several clustering models have been extensively studied, such as $k$-center, $k$-median, and $k$-means clustering~\cite{awasthi2014center}.
%
%
In reality, datasets often are noisy and contain outliers. Moreover, outliers could seriously affect the final results in data analysis~\cite{tan2006introduction,chandola2009anomaly}. Clustering with outliers can be viewed as a generalization of ordinary clustering problems; however, the existence of outliers makes the problems to be much more challenging. 


We focus on the problem of {\em $k$-center clustering with outliers} in this paper. Given a metric space with $n$ vertices and a pre-specified number of outliers $z<n$, the problem is to find $k$ balls to cover at least $n-z$ vertices and minimize the maximum radius of the balls. The problem also can be defined in Euclidean space so that the cluster centers can be any points in the space (i.e., not restricted to be selected from the input points). The $2$-approximation algorithms for ordinary $k$-center clustering (without outliers) were given in \cite{gonzalez1985clustering,hochbaum1985best}, and it was proved that any approximation ratio lower than $2$ implies $P=NP$. 
A $3$-approximation algorithm for $k$-center clustering with outliers in arbitrary metrics was proposed by Charikar et al.~\cite{charikar2001algorithms}; for the problem in Euclidean space, their approximation ratio becomes $4$. A following streaming $(4+\epsilon)$-approximation algorithm was proposed by McCutchen and Khuller~\cite{mccutchen2008streaming}. 
Recently, Chakrabarty et al.~\cite{DBLP:conf/icalp/ChakrabartyGK16} proposed a $2$-approximation algorithm for metric $k$-center clustering with outliers (but it is unclear of the resulting approximation ratio for the problem in Euclidean space). 
Existing algorithms often have high time complexities. For example, the complexities of the algorithms in~\cite{charikar2001algorithms,mccutchen2008streaming} are $O(k n^2\log n)$ and $O\big(\frac{1}{\epsilon}(kzn+(kz)^2\log \Phi)\big)$ respectively, where $\Phi$ is the ratio of the optimal radius to the smallest pairwise distance among the vertices; the algorithm in~\cite{DBLP:conf/icalp/ChakrabartyGK16} needs to solve a complicated model of linear programming and the exact time complexity is not provided. 
The coreset based idea of Badoiu et al.~\cite{BHI} needs to enumerate a large number of possible cases and also yields a high complexity. Several distributed algorithms for $k$-center clustering with outliers were proposed recently~\cite{malkomes2015fast,guha2017distributed,DBLP:journals/corr/abs-1802-09205,li2018distributed}; most of these distributed algorithms, to our best knowledge, rely on the sequential algorithm~\cite{charikar2001algorithms}.


%

In this paper, we aim to design quality guaranteed algorithm with low complexity for the problem of $k$-center clustering with outliers. 
Our idea is inspired by the greedy method from Gonzalez~\cite{gonzalez1985clustering} for solving ordinary $k$-center clustering. Based on some novel insights, we show that this greedy method also works for the problem with outliers (Section~\ref{sec-center}). Our approach can achieve the approximation ratio $2$ with respect to the clustering cost (i.e., the radius); moreover, the time complexity is linear in the input size. 
Charikar et al.\cite{charikar2003better} showed that if more than $z$ outliers are allowed to remove, the random sampling technique can be applied to reduce the data size for metric $k$-center clustering with outliers. Recently, Huang et al.\cite{huang2018epsilon} showed a similar result for instances in Euclidean space (and they name the sample as ``robust coreset''). In Section~\ref{sec-core1}, we prove that the sample size of \cite{huang2018epsilon} can be further reduced.

%

We also consider the problem in doubling metrics, motivated by the fact that many real-world datasets often manifest low intrinsic dimensions~\cite{belkin2003problems}. 
For example, image sets usually can be represented in low dimensional manifold though the Euclidean dimension of the image vectors can be very high. ``Doubling dimension'' is widely used for measuring the intrinsic dimensions of datasets~\cite{talwar2004bypassing} (the formal definition is given in Section~\ref{sec-pre}). Rather than assuming the whole $(X,d)$ has a low doubling dimension,  we only assume that \textbf{the inliers of the given data have a low doubling dimension $\rho>0$.} We do not have any assumption on the outliers; namely, the outliers can scatter arbitrarily in the space. We believe that this assumption captures a large range of high dimensional instances in reality. 

With the assumption, we show that our approach can further improve the clustering quality. In particular, the greedy approach is able to construct a coreset for the problem of $k$-center clustering with outliers; as a consequence, the time complexity can be significantly reduced if running existing algorithms on the coreset (Section~\ref{sec-doubling}). {\em coreset} construction is a technique for reducing data size so as to speedup many optimization problems; we refer the reader to the surveys~\cite{DBLP:journals/corr/Phillips16,bachem2017practical} for more details. The size of our coreset is $2z+O\big((2/\mu)^\rho k\big)$, where $\mu$ is a small parameter measuring the quality of the coreset; the construction time is $O((\frac{2}{\mu})^\rho kn)$. Note that $z$ and $k$ are often much smaller than $n$ in practice; the coefficient $2$ of $z$ actually can be further reduced to be arbitrarily close to $1$, by increasing the coefficient of the second term $(2/\mu)^\rho k$. Moreover, our coreset is a natural ``composable coreset''~\cite{DBLP:conf/pods/IndykMMM14} which could be potentially applied to distributed clustering with outliers. Very recently, Ceccarello et al.\cite{DBLP:journals/corr/abs-1802-09205} also provided a coreset for $k$-center clustering with $z$ outliers in doubling metrics, where their size is $T=O((k+z)(\frac{24}{\mu})^\rho)$ with $O(nT )$ construction time. Thus our result is a significant improvement in terms of coreset size and construction time. 
Huang et al.\cite{huang2018epsilon} considered the coreset construction for $k$-median/means clustering with outliers in doubling metrics, however, their method cannot be extended to the case of $k$-center. Aghamolaei and  Ghodsi \cite{DBLP:conf/cccg/AghamolaeiG18} considered the coreset construction for ordinary $k$-center clustering without outliers.

Our proposed algorithms are easy to implement in practice. To study the performance of our algorithms, we test them on both synthetic and real datasets in Section~\ref{sec-exp}. The experimental results suggest that our method outperforms existing methods in terms of clustering quality and running time. Also, the running time can be significantly reduced via building coreset where the clustering quality can be well preserved simultaneously.

\vspace{-0.1in}

\subsection{Preliminaries}
\label{sec-pre}
\vspace{-0.05in}

We consider the problem of $k$-center with outliers in arbitrary metrics and Euclidean space $\mathbb{R}^D$. Let $(X, d)$ be a metric, where $X$ contains $n$ vertices and $d(\cdot, \cdot)$ is the distance function; with a slight abuse of notation, we also use the function $d$ to denote the shortest distance between two subsets $X_1, X_2\subseteq X$, i.e., $d(X_1, X_2)=\min_{p\in X_1, q\in X_2}d(p, q)$. We assume that the distance between any pair of vertices in $X$ is given in advance; for the problem in Euclidean space, it takes $O(D)$ time to compute the distance between any pair of points. 
Below, we introduce several important definitions that are used throughout the paper.

\vspace{-0.05in}

\begin{definition}[$k$-Center Clustering with Outliers]
\label{def-outlier}
Given a metric $(X, d)$ with two positive integers $k$ and $z<n$, $k$-center clustering with outliers is to find a subset $X'\subseteq X$, where $|X'|\geq n-z$, and $k$ centers $\{c_1, \cdots, c_k\}\subseteq X$, such that $\max_{p\in X'}\min_{1\leq j\leq k}d(p, c_j)$ is minimized. If given a set $P$ of  $n$ points in $\mathbb{R}^D$, the problem is to find a subset $P'\subseteq P$, where $|P'|\geq n-z$, and $k$ centers $\{c_1, \cdots, c_k\}\subseteq\mathbb{R}^D$, such that $\max_{p\in P'}\min_{1\leq j\leq k}||p-c_j||$ is minimized.
%
\end{definition}

\textbf{Note.} For the sake of convenience, we describe the following definitions only in terms of metric space. In fact, the definitions can be easily modified for the problem in Euclidean space.

In this paper, we always use $X_{opt}$, a subset of $X$ with size $n-z$, to denote the subset yielding the optimal solution. Also, let $\{C_1, \cdots, C_k\}$ be the $k$ clusters forming $X_{opt}$, and the resulting clustering cost be $r_{opt}$; that is, each $C_j$ is covered by an individual ball with radius $r_{opt}$.

Usually, optimization problems with outliers are challenging to solve. Thus we often relax our goal and allow to miss a little more outliers in practice. Actually the same relaxation idea has been adopted by a number of works on clustering with outliers before~\cite{charikar2003better,huang2018epsilon,alon2003testing,li2018distributed}.

\vspace{-0.05in}
\begin{definition}[$(k,z)_{\epsilon}$-Center Clustering]
\label{def-relax}
Let $(X,d)$ be an instance of $k$-center clustering with $z$ outliers, and $\epsilon\geq 0$. $(k,z)_{\epsilon}$-center clustering is to find a subset $X'$ of $X$, where $|X'|\geq n-(1+\epsilon)z$, such that the corresponding clustering cost of Definition~\ref{def-outlier} on $X'$ is minimized.

\textbf{(\rmnum{1})} Given a set $A$ of cluster centers ($|A|$ could be larger than $k$), the resulting clustering cost, 
\begin{eqnarray}
 \min\big\{\max_{p\in X'}\min_{c\in A}d(p, c)\mid X'\subseteq X, |X'|\geq n-(1+\epsilon)z\big\}
\end{eqnarray}
is denoted by $\phi_{\epsilon}(X, A)$.

\textbf{(\rmnum{2})} If $|A|=k$ and $\phi_{\epsilon}(X, A)\leq\alpha r_{opt}$ with $\alpha>0$\footnote{Since we remove more than $z$ outliers, it is possible to have an approximation ratio $\alpha<1$, i.e, $\phi_{\epsilon}(X, A)< r_{opt}$.}, 
 it is called an $\alpha$-approximation. Moreover, if $|A|=\beta k$ with $\beta> 1$, it is called an $(\alpha, \beta)$-approximation.

\end{definition}


Obviously, the problem in Definition~\ref{def-outlier} is a special case of $(k,z)_{\epsilon}$-center clustering with $\epsilon=0$. 
Further, Definition~\ref{def-outlier} and \ref{def-relax} can be naturally extended to \textbf{weighted case:} each vertex $p$ has a non-negative weight $w_p$ and the total weight of outliers should be equal to $z$; the distance $d(p, c_j)$ in the objective function is replaced by $w_p\cdot d(p, c_j)$. 
Then we have the following definition of coreset.
\vspace{-0.03in}
\begin{definition}[Coreset]
\label{def-coreset}
Given a small parameter $\mu\in(0,1)$ and an instance $(X,d)$ of  $k$-center clustering with $z$ outliers, a set $S\subseteq X$ is called a $\mu$-coreset of $X$, if each vertex of $S$ is assigned a non-negative weight and $\phi_{0}(S, H)\in (1\pm\mu)\phi_{0}(X, H)$ for any set $H\subseteq X$ of $k$ vertices.
\end{definition}
\vspace{-0.03in}

Given a large-scale instance $(X, d)$, we can run existing algorithm on its coreset $S$ to compute an approximate solution for $X$; if $|S|\ll n$, the resulting running time can be significantly reduced. Formally, we have the following claim (see the proof in Section~\ref{sec-proof-c1}).

\begin{claim}
\label{pro-core}
If the set $H$ yields an $\alpha$-approximation of the $\mu$-coreset $S$, it yields an $\alpha\times\frac{1+\mu}{1-\mu}$-approximation of $X$.
\end{claim}

As mentioned before, we also consider the case with low doubling dimension. Roughly speaking, doubling dimension describes the expansion rate of the metric. 
For any $p\in X$ and $r\geq 0$, we use $Ball(p, r)$ to denote the ball centered at $p$ with radius $r$.
\vspace{-0.03in}

\begin{definition}[Doubling Dimension]
\label{def-dd}
The doubling dimension of a metric $(X,d)$ is the smallest number $\rho>0$, such that for any $p\in X$ and $r\geq 0$, $X\cap Ball(p, 2r)$ is always covered by the union of at most $2^\rho$ balls with radius $r$.
\end{definition}
\vspace{-0.05in}
\section{Algorithms for $(k,z)_{\epsilon}$-Center Clustering}
\label{sec-center}
\vspace{-0.05in}

For the sake of completeness, let us briefly introduce the algorithm of \cite{gonzalez1985clustering} for ordinary $k$-center clustering first. Initially, it arbitrarily selects a vertex from $X$, and iteratively selects the following $k-1$ vertices, where each $j$-th step ($2\leq j\leq k$) chooses the vertex having the largest minimum distance to the already selected $j-1$ vertices; finally, each input vertex is assigned to its nearest neighbor of these selected $k$ vertices. It can be proved that this greedy strategy results in a $2$-approximation of $k$-center clustering; the algorithm also works for the problem in Euclidean space and results in the same approximation ratio. In this section, we show that a modified version of Gonzalez's algorithm yields approximate solutions for $(k,z)_{\epsilon}$-center clustering. 

In Section~\ref{sec-center-bi} and \ref{sec-center-single}, we present our results for metric $k$-center with outliers. Actually, it is easy to see that Algorithm \ref{alg-bi} and \ref{alg-single} yield the same approximation ratios if the input instance is a set of points in Euclidean space (the analysis is almost identical, and we omit the details due to the space limit); only the running times are different, since  it takes $O(D)$ time to compute distance between two points in $\mathbb{R}^D$.

%

\vspace{-0.03in}
\subsection{$(2, O(\frac{1}{\epsilon}))$-Approximation}
\label{sec-center-bi}
\vspace{-0.03in}

Here, we consider bi-criteria approximation that returns more than $k$ cluster centers. The main challenge for implementing Gonzalez's algorithm is that the outliers and inliers are mixed in $X$; for example, the selected vertex, which has the largest minimum distance to the already selected vertices, is very likely to be an outlier, and therefore the clustering quality could be arbitrarily bad. Instead, our strategy is to take a small sample from the farthest subset. We implement our idea in Algorithm~\ref{alg-bi}. For simplicity, let $\gamma$ denote $z/n$ in the algorithm; usually we can assume that $\gamma$ is a value much smaller than $1$. We prove the correctness of Algorithm~\ref{alg-bi} below.


\begin{algorithm}[tb]
   \caption{Bi-criteria Approximation Algorithm}
   \label{alg-bi}
\begin{algorithmic}
  \STATE {\bfseries Input:} An instance $(X, d)$ of metric $k$-center clustering with $z$ outliers, and $|X|=n$; parameters $\epsilon>0$, $\eta\in (0,1)$, and $t\in\mathbb{Z}^+$.
   \STATE
\begin{enumerate}
\item Let $\gamma=z/n$ and initialize a set $E=\emptyset$. 

\item Initially, $j=1$; randomly select $\frac{1}{1-\gamma}\log\frac{1}{\eta}$ vertices from $X$ and add them to $E$.
\item Run the following steps until $j= t$:
\begin{enumerate}
\item $j=j+1$ and let $Q_j$ be the farthest $(1+\epsilon)z$ vertices of $X$ to $E$ (for each vertex $p\in X$, its distance to $E$ is $\min_{q\in E}d(p, q)$). 
\item Randomly select $\frac{1+\epsilon}{\epsilon}\log\frac{1}{\eta}$ vertices from $Q_j$ and add them to $E$.
\end{enumerate}
%
%
%
%
\end{enumerate}
  \STATE {\bfseries Output} $E$.
\end{algorithmic}
\end{algorithm}




\vspace{-0.05in}
\begin{lemma}
\label{lem-select1}
With probability at least $1-\eta$, the set $E$ in Step 2 of Algorithm~\ref{alg-bi} contains at least one point from $X_{opt}$.
\end{lemma}

Since $|X_{opt}|/|X|= 1-\gamma$, Lemma~\ref{lem-select1} can be easily obtained by the following folklore claim (we show the proof in Section~\ref{sec-pro-sample}). 
\vspace{-0.05in}
\begin{claim}
\label{pro-sample}
Let $U$ be a set of elements and $V\subseteq U$ with $\frac{|V|}{|U|}=\tau>0$. Given $\eta\in(0,1)$, if one randomly samples $\frac{1}{\tau}\log\frac{1}{\eta}$ elements from $U$, with probability at least $1-\eta$, the sample contains at least one element from $V$.
\end{claim}
\vspace{-0.05in}

Recall that $\{C_1, C_2, \cdots, C_k\}$ are the $k$ clusters forming $X_{opt}$. Denote by $\lambda_j(E)$ the number of the clusters which have non-empty intersection with $E$ at the beginning of $j$-th round in Step~3 of Algorithm~\ref{alg-bi}. For example, initially $\lambda_1(E)\geq 1$ by Lemma~\ref{lem-select1}. Obviously, if $\lambda_j(E)=k$, i.e., $C_l\cap E\neq\emptyset$ for any $1\leq l\leq k$, $E$ yields a $2$-approximation for $k$-center clustering with outliers through the triangle inequality.
\vspace{-0.02in}
\begin{claim}
\label{cla-e2}
If $\lambda_j(E)=k$,  then $\phi_0(X, E)\leq 2 r_{opt}$.
\end{claim}
\vspace{-0.05in}

\begin{lemma}
\label{lem-select2}
In each round of Step 3 of Algorithm~\ref{alg-bi}, with probability at least $1-\eta$, either (1)  $d(Q_j, E)\leq 2 r_{opt}$ or (2) $\lambda_j(E)\geq\lambda_{j-1}(E)+1$.
\end{lemma}
\begin{proof}
Suppose that (1) is not true, i.e., $d(Q_j, E)> 2 r_{opt}$, and we prove that (2) is true. Let $\mathcal{J}$ include all the indices $l\in\{1, 2, \cdots, k\}$ with $ E\cap C_l\neq\emptyset$. We claim that $Q_j\cap C_l=\emptyset$ for each $l\in \mathcal{J}$. Otherwise, let $p\in Q_j\cap C_l$ and $p'\in E\cap C_l$; due to the triangle inequality, we know that $d(p,p')\leq 2 r_{opt}$ which is in contradiction to the assumption $d(Q_j, E)> 2 r_{opt}$. Thus, $Q_j\cap X_{opt}$ only contains the vertices from $C_l$ with $l\notin \mathcal{J}$. Moreover, since the number of outliers is $z$, we know that $\frac{|Q_j\cap X_{opt}|}{|Q_j|}\geq \frac{\epsilon}{1+\epsilon}$. By Claim~\ref{pro-sample}, if randomly selecting $\frac{1+\epsilon}{\epsilon}\log\frac{1}{\eta}$ vertices from $Q_j$, with probability at least $1-\eta$, the sample contains at least one vertex from $Q_j\cap X_{opt}$; also, the vertex must come from $\cup_{l\notin \mathcal{J}}C_l$. That is, (2) $\lambda_j(E)\geq\lambda_{j-1}(E)+1$ happens.
\qed\end{proof}

If (1) of Lemma~\ref{lem-select2} happens, i.e., $d(Q_j, E)\leq 2 r_{opt}$, then it implies that $\max_{p\in X\setminus Q_j}d(p, E)\leq 2 r_{opt}$; 
moreover, since $|Q_j|=(1+\epsilon)z$, we have $\phi_{\epsilon}(X,E)\leq 2 r_{opt}$. 
Next, we assume that (1) in Lemma~\ref{lem-select2} never happens, and prove that $\lambda_j(E)=k$ with constant probability when $j=\Theta(k)$. The following idea actually has been used by Aggarwal et al.~\cite{aggarwal2009adaptive} for achieving a bi-criteria approximation for $k$-means clustering. Define a random variable $x_j$: $x_j=1$ if $\lambda_j(E)=\lambda_{j-1}(E)$, or $0$ if $\lambda_j(E)\geq\lambda_{j-1}(E)+1$, for $j=1, 2, \cdots$. So $\mathbb{E}[x_j]\leq\eta$ by Lemma~\ref{lem-select2} and
\begin{eqnarray}
\sum_{1\leq s\leq j}(1-x_s)\leq\lambda_j(E). \label{for-azuma2}
\end{eqnarray}
Also, let $J_j=\sum_{1\leq s\leq j}(x_s-\eta)$ and $J_0=0$. Then, $\{J_0, J_1, J_2, \cdots\}$ is a super-martingale with $J_{j+1}-J_j< 1$ (more details are shown in Section~\ref{sec-proof-the}). Through {\em Azuma-Hoeffding inequality}~\cite{alon2004probabilistic}, we have 
$Pr(J_t\geq J_0+\delta)\leq e^{-\frac{\delta^2}{2t}}$ 
for any $t\in\mathbb{Z}^+$ and $\delta>0$. Let $t=\frac{k+\sqrt{k}}{1-\eta}$ and $\delta=\sqrt{k}$, the inequality implies that
\begin{eqnarray}
Pr(\sum_{1\leq s\leq t}(1-x_s)\geq k)\geq 1-e^{-\frac{1-\eta}{4}}. \label{for-azuma}
\end{eqnarray}
Combining (\ref{for-azuma2}) and (\ref{for-azuma}), we know that $\lambda_t(E)=k$ with probability at least $1-e^{-\frac{1-\eta}{4}}$. Moreover, $\lambda_t(E)=k$ directly implies that $E$ is a $2$-approximate solution by Claim~\ref{cla-e2}. Together with Lemma~\ref{lem-select1}, we have the following theorem.

\begin{theorem}
\label{the-biapprox}
Let $\epsilon>0$. If we set $t=\frac{k+\sqrt{k}}{1-\eta}$ for  Algorithm~\ref{alg-bi}, with probability at least $(1-\eta)(1-e^{-\frac{1-\eta}{4}})$, $\phi_{\epsilon}(X,E)\leq 2 r_{opt}$. 
\end{theorem}

\textbf{Quality and Running time.} If $\frac{1}{\eta}$ and $\frac{1}{1-\gamma}$ are constant numbers, Theorem~\ref{the-biapprox} implies that  $E$ is a $\big(2, O(\frac{1}{\epsilon})\big)$-approximation for  $(k,z)_{\epsilon}$-center clustering of $X$ with constant probability. In each round of Step~3, there are $O(\frac{1}{\epsilon})$ new vertices added to $E$, thus it takes $O(\frac{1}{\epsilon}n)$ time to update the distances from the vertices of $X$ to $E$; to select the set $Q_j$, we can apply the linear time selection algorithm~\cite{blum1973time}. Overall, the running time of Algorithm~\ref{alg-bi} is $O(\frac{k}{\epsilon}n)$. If the given instance is in $\mathbb{R}^D$, the running time will be $O(\frac{k}{\epsilon}n D)$.

Further, we consider the instances under some practical assumption, and provide new analysis of Algorithm \ref{alg-bi}. 
In reality, the clusters are usually not too small, compared with the number of outliers. For example, it is rare to have a cluster $C_l$ that $|C_l|\ll z$. 


\begin{theorem}
\label{the-biapprox2}
If each optimal cluster $C_l$ has size at least $\epsilon z$ for $1\leq l\leq k$, the set $E$ of Algorithm~\ref{alg-bi} is a $\big(4, O(\frac{1}{\epsilon})\big)$-approximation for the problem of $(k,z)_{0}$-center clustering with constant probability. 
\end{theorem}
Compared with Theorem~\ref{the-biapprox}, Theorem~\ref{the-biapprox2} shows that we can exactly exclude $z$ outliers (rather than $(1+\epsilon)z$), though the approximation ratio with respect to the radius becomes $4$. 



\begin{proof}[Proof of Theorem~\ref{the-biapprox2}]
We take a more careful analysis on the proof of Lemma~\ref{lem-select2}. If (1) never happens, eventually $\lambda_j(E)$ will reach $k$ and thus $\phi_0(X, E)\leq 2 r_{opt}$ (Claim~\ref{cla-e2}). So we focus on the case that (1) happens before $\lambda_j(E)$ reaching $k$. 
Suppose at $j$-th round, $d(Q_j, E)\leq 2 r_{opt}$ but $\lambda_j(E)< k$. 
We consider two cases \textbf{(\rmnum{1})} there exists some $l_0\notin \mathcal{J}$ such that $C_{l_0}\subseteq Q_j$ and \textbf{(\rmnum{2})} otherwise.

For \textbf{(\rmnum{1})}, we have $C_{l_0}\subseteq Q_j$ for some $l_0\notin \mathcal{J}$. Note that we assume $|C_{l_0}|\geq \epsilon z$, i.e., $\frac{|C_{l_0}|}{|Q_j|}\geq \frac{\epsilon}{1+\epsilon}$. Using the same manner in the proof of Lemma~\ref{lem-select2}, we know that (2) $\lambda_j(E)\geq\lambda_{j-1}(E)+1$ happens with probability $1-\eta$. Thus, if  \textbf{(\rmnum{1})} is always true, we can continue Step 3 and eventually $\lambda_j(E)$ will reach $k$, that is,  a $\big(2, O(\frac{1}{\epsilon})\big)$-approximation of $(k,z)_{0}$-center clustering is obtained with constant probability.

For \textbf{(\rmnum{2})}, we have $C_{l}\setminus Q_j\neq\emptyset$ for all $l\notin\mathcal{J}$. Together with the assumption $d(Q_j, E)\leq 2 r_{opt}$, we know that there exists $q_l\in C_{l}\setminus Q_j$ (for each $l\notin\mathcal{J}$) such that $d(q_l, E)\leq d(Q_j, E)\leq 2 r_{opt}$. Consequently, we have that $\forall q\in C_{l}$, 
\begin{eqnarray}
d(q, E)&\leq& ||q-q_l||+d(q_l, E)\leq 4 r_{opt} \text{ (see the left of Figure.~\ref{fig-4app})}.
\end{eqnarray}
Note that for any $l\in \mathcal{J}$, $d(E, C_l)\leq 2 r_{opt}$ by the triangle inequality. Thus, 
\begin{eqnarray}
\phi_0(X, E)&\leq& \max_{q\in \cup^k_{l=1}C_l} d(q, E)\leq 4 r_{opt}. \label{for-biapprox21}
\end{eqnarray}
 So a $\big(4, O(\frac{1}{\epsilon})\big)$-approximation of $(k,z)_{0}$-center clustering is obtained.
\qed\end{proof}


\begin{algorithm}[tb]
   \caption{$2$-Approximation Algorithm}
   \label{alg-single}
\begin{algorithmic}
  \STATE {\bfseries Input:} An instance $(X,d)$ of metric $k$-center clustering with $z$ outliers, and $|X|=n$; a parameter $\epsilon>0$.
   \STATE
\begin{enumerate}
\item Initialize a set $E=\emptyset$.

\item Let $j=1$; randomly select one vertex from $X$ and add it to $E$.
\item Run the following steps until $j= k$:
\begin{enumerate}
\item $j=j+1$ and let $Q_j$ be the farthest $(1+\epsilon)z$ vertices to $E$. 
\item Randomly select one vertex from $Q_j$ and add it to $E$.
\end{enumerate}
%
%
%
%
\end{enumerate}
  \STATE {\bfseries Output} $E$.
\end{algorithmic}
\end{algorithm}

\subsection{$2$-Approximation}
\label{sec-center-single}
If $k$ is a constant, we show that a single-criterion $2$-approximation can be achieved. Actually, we use the same strategy as Section~\ref{sec-center-bi}, but run only $k$ rounds with each round sampling only one vertex. See Algorithm~\ref{alg-single}. 


\begin{wrapfigure}{r}{0.5\textwidth}
  \vspace{-25pt}
\begin{center}
    \includegraphics[width=0.45\textwidth]{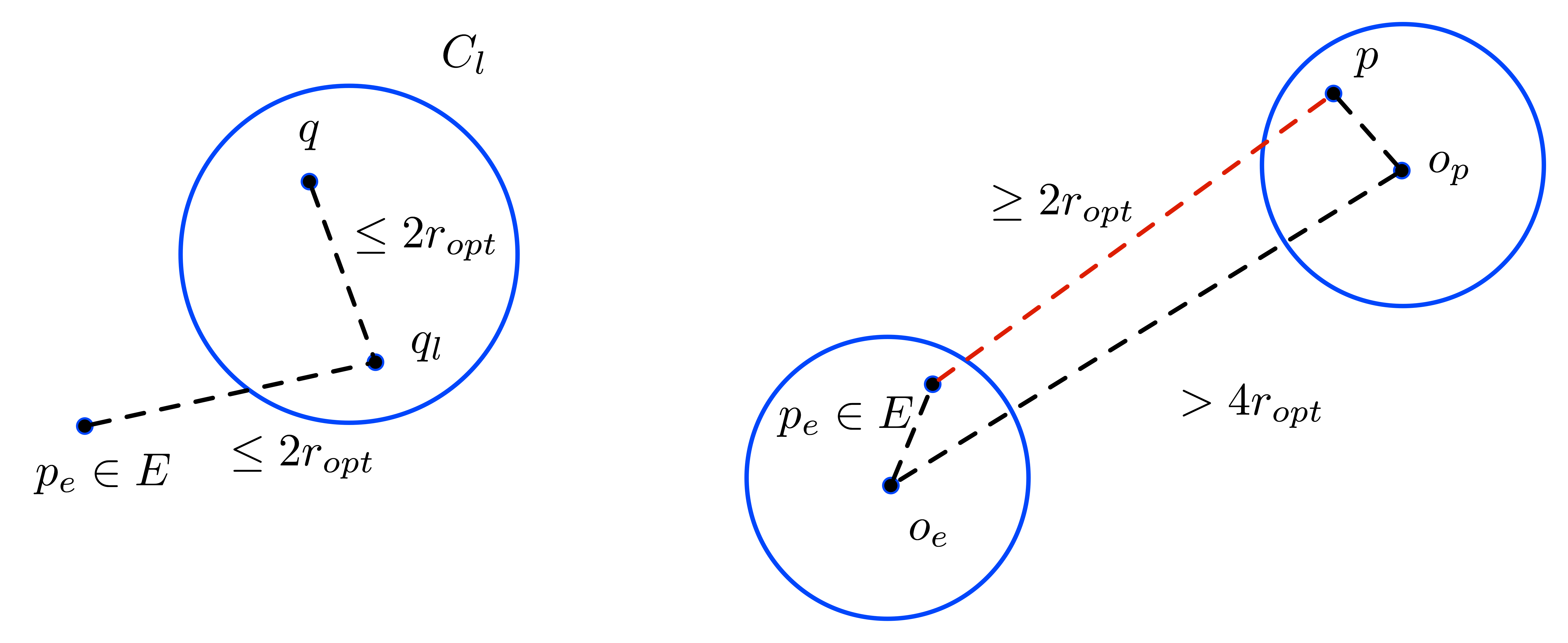}  
    \end{center}
  \vspace{-22pt}
  \caption{Left: $p_e$ is a point of $E$ having distance  $\leq 2 r_{opt}$ to $p_l$; right: $p_e$ is any point of $E$, $o_e$ and $o_p$ are the centers taking charge of $p_e$ and $p$.}   
   \label{fig-4app}
     \vspace{-15pt}
\end{wrapfigure}

Denote by $\{v_1, \cdots, v_k\}$ the $k$ sampled vertices of $E$. Actually, the proof of Theorem~\ref{the-kcenter} is similar to the analysis in Section~\ref{sec-center-bi}. The only difference is that the probability that (2) $\lambda_j(E)\geq\lambda_{j-1}(E)+1$ happens is at least $\frac{\epsilon}{1+\epsilon}$. Also note that $v_1\in X_{opt}$ with probability $1-\gamma$ ($\gamma=z/n$). If all of these events happen, either we obtain a $2$-approximation before $k$ steps (i.e., $d(E, X\setminus Q_j)\leq 2 r_{opt}$ for some $j<k$), or $\{v_1, \cdots, v_k\}$ fall into the $k$ optimal clusters $C_1, C_2, \cdots, C_k$ separately (i.e., $\lambda_k(E)=k$). No matter which case happens, we always obtain a $2$-approximation with respect to $(k,z)_{\epsilon}$-center clustering. So we have Theorem~\ref{the-kcenter}.

\begin{theorem}
\label{the-kcenter}
With probability at least $(1-\gamma)(\frac{\epsilon}{1+\epsilon})^{k-1}$, Algorithm~\ref{alg-single} returns a $2$-approximation for the problem of $(k,z)_{\epsilon}$-center clustering on $X$. The running time is $O(kn)$. If the given instance is in $\mathbb{R}^D$, the running time will be $O(kn D)$.
\end{theorem}

To boost the probability of Theorem~\ref{the-kcenter}, we just need to repeatedly run the algorithm; the success probability is easy to calculate by taking the union bound.


\begin{corollary}
\label{the-kcenter2}
If we run Algorithm~\ref{alg-single} $ O\big(\frac{1}{1-\gamma}(\frac{1+\epsilon}{\epsilon})^{k-1}\big)$ times, with constant probability, at least one time the algorithm  returns a $2$-approximation for the problem of $(k,z)_\epsilon$-center clustering.
\end{corollary}


Similar to Theorem~\ref{the-biapprox2}, we consider the practical instances. 
We show that the quality of Theorem~\ref{the-kcenter} can be preserved even exactly excluding $z$ outliers, if the optimal clusters are ``well separated''. The property was also studied for other clustering problems in practice~\cite{DBLP:journals/pami/KanungoMNPSW02,DBLP:journals/corr/abs-1205-4891}. Let $\{o_1, \cdots, o_k\}$ be the $k$ cluster centers of the optimal clusters $\{C_1, \cdots, C_k\}$.

\begin{theorem}
\label{the-kcenter3}
Suppose that each optimal cluster $C_l$ has size at least $\epsilon z$ and $||o_l-o_{l'}||>4 r_{opt}$ for $1\leq l\neq l'\leq k$. Then with probability at least $(1-\gamma)(\frac{\epsilon}{1+\epsilon})^{k-1}$, Algorithm~\ref{alg-single} returns a $2$-approximation for the problem of $(k,z)_0$-center clustering.
\end{theorem}
\begin{proof}
Initially, we know that $\lambda_1(E)=1$ with probability $1-\gamma$. Suppose that at the beginning of the $j$-th round of Algorithm~\ref{alg-single} with $2\leq j\leq k$, $E$ already has $j-1$ vertices separately falling in $j-1$ optimal clusters; also, we still let $\mathcal{J}$ be the set of the indices of these $j-1$ clusters. Then we have the following claim.


\begin{claim}
\label{cl-kcenter3}
$|Q_j\cap (\cup_{l\notin \mathcal{J}}C_l)|\geq \epsilon z$.
\end{claim}
\begin{proof}
For any $p\in  \cup_{l\notin \mathcal{J}}C_l$, we have 
\begin{eqnarray}
d(p, E)> 4 r_{opt}- r_{opt}-r_{opt}=2 r_{opt} \label{for-cl1}
\end{eqnarray}
from triangle inequality and the assumption $||o_l-o_{l'}||>4 r_{opt}$ for $1\leq l\neq l'\leq k$ (see the right of Figure.~\ref{fig-4app}). In addition, for any $p\in  \cup_{l\in \mathcal{J}}C_l$, we have 
\begin{eqnarray}
d(p, E)\leq 2 r_{opt}. \label{for-cl2}
\end{eqnarray}

We consider two cases. If $d(Q_j, E)\leq 2 r_{opt}$ at the current round, then (\ref{for-cl1}) directly implies that $\cup_{l\notin \mathcal{J}}C_l\subseteq Q_j$ (recall $Q_j$ is the set of farthest vertices to $E$); thus $|Q_j\cap (\cup_{l\notin \mathcal{J}}C_l)|=|\cup_{l\notin \mathcal{J}}C_l|\geq \epsilon z$ by the assumption that any $|C_l|\geq \epsilon z$. 
Otherwise, $d(Q_j, E)> 2 r_{opt}$. Then $Q_j\cap (\cup_{l\in \mathcal{J}}C_l)=\emptyset$ by (\ref{for-cl2}). Moreover, since there are only $z$ outliers and $|Q_j|=(1+\epsilon)z$, we know that $|Q_j\cap (\cup_{l\notin \mathcal{J}}C_l)|\geq \epsilon z$.
\qed\end{proof}

Claim~\ref{cl-kcenter3} reveals that with probability at least $\frac{\epsilon}{1+\epsilon}$, the new added vertex falls in $\cup_{l\notin \mathcal{J}}C_l$, i.e., $\lambda_j(E)=\lambda_{j-1}(E)+1$. Overall, we know that $\lambda_k(E)=k$, i.e., $E$ is a $2$-approximation of $(k, z)_0$-center clustering (by Claim~\ref{cla-e2}),  with probability at least $(1-\gamma)(\frac{\epsilon}{1+\epsilon})^{k-1}$. 
\qed\end{proof}

\subsection{Reducing Data Size via Random Sampling}
\label{sec-core1}

Given a metric $(X,d)$, Charikar et al.~\cite{charikar2003better} showed that we can use a random sample $S$ to replace $X$. Recall $\gamma=z/n$. Let $|S|=O(\frac{k}{\epsilon^2 \gamma}\ln n)$ and $E$ be an $\alpha$-approximate solution of $(k, z)_\epsilon$-center clustering on $(S,d)$, then $E$ is an $\alpha$-approximate solution of $(k, z)_{O(\epsilon)}$-center clustering on $(X,d)$ with constant probability. In $D$-dimensional Euclidean space, Huang et al.~\cite{huang2018epsilon} showed a similar result, where the sample size $|S|=\tilde{O}(\frac{1}{\epsilon^2\gamma^2}kD)$\footnote{The asymptotic notation $\tilde{O}(f)=O\big(f\cdot polylog(\frac{kD}{\epsilon\gamma})\big)$.} (to be consistent with our paper, we change the notations in their theorem). In this section, we show that the sample size of \cite{huang2018epsilon} can be further improved to be $\tilde{O}(\frac{1}{\epsilon^2\gamma}kD)$, which can be a significant improvement if $\frac{1}{\gamma}=\frac{n}{z}$ is large.

Let $P$ be a set of $n$ points in $\mathbb{R}^D$. Consider the range space $\Sigma=(P, \Pi)$ where each range $\pi\in \Pi$ is the complement of union of $k$ balls in $\mathbb{R}^D$. We know that the VC dimension of balls is $O(D)$~\cite{alon2004probabilistic}, and therefore the VC dimension of union of $k$ balls is $O(kD \log k)$~\cite{blumer1989learnability}. That is, the VC dimension of the range space $\Sigma$ is $O(kD \log k)$.
Let $\epsilon\in(0,1)$, and an ``$\epsilon$-sample'' $S$ of $P$ is defined as follows: $\forall \pi\in\Pi$, $\big|\frac{|\pi\cap P|}{|P|}-\frac{|\pi\cap S|}{|S|}\big|\leq \epsilon$; roughly speaking, $S$ is an approximation of $P$ with an additive error inside each range $\pi$. 
Given a range space with VC dimension $m$, an $\epsilon$-sample can be easily obtained via uniform sampling~\cite{alon2004probabilistic}, where the success probability is $1-\lambda$ and the sample size is $O\big(\frac{1}{\epsilon^2}(m\log\frac{m}{\epsilon}+\log\frac{1}{\lambda})\big)$ for any $0<\lambda<1$.
For our problem, we need to replace the ``$\epsilon$'' of the ``$\epsilon$-sample'' by $\epsilon\gamma$ to guarantee that the number of uncovered points is bounded by $\big(1+O(\epsilon)\big)\gamma n$ (we show the details below); the resulting sample size will be $\tilde{O}(\frac{1}{\epsilon^2\gamma^2}kD)$ that is the same as the sample size of  \cite{huang2018epsilon} (we assume that the term $\log\frac{1}{\lambda}$ is a constant for convenience). 

%

Actually, the front factor $\frac{1}{\epsilon^2\gamma^2}$ of the sample size can be further reduced to be $\frac{1}{\epsilon^2\gamma}$ by a more careful analysis. We observe that there is no need to guarantee the additive error for each range $\pi$ (as the definition of $\epsilon$-sample). Instead, only a multiplicative error for the ranges covering at least $\gamma n$ points should be sufficient. Note that when a range covers more points, the multiplicative error is weaker than the additive error and thus the sample size is reduced. For this purpose, we use {\em relative approximation}~\cite{har2011relative,li2001improved}: let $S\subseteq P$ be a subset of size $\tilde{O}(\frac{1}{\epsilon^2\gamma}kD)$ chosen uniformly at random, then with constant probability,
\begin{eqnarray}
\forall \pi\in\Pi,\ \Big|\frac{|\pi\cap P|}{|P|}-\frac{|\pi\cap S|}{|S|}\Big|\leq \epsilon\times\max\Big\{\frac{|\pi\cap P|}{|P|}, \gamma\Big\}. \label{for-relativesample}
\end{eqnarray}
We formally state our result below.

\begin{theorem}
\label{the-samplereduce}
Let $P$ be an instance for the problem of $k$-center clustering with outliers in $\mathbb{R}^{D}$ as described in Definition~\ref{def-outlier}, and $S\subseteq P$ be a subset of size $\tilde{O}(\frac{1}{\epsilon^2\gamma}kD)$ chosen uniformly at random. Suppose $\epsilon\leq 0.5$. Let $S$ be a new instance for the problem of $k$-center clustering with outliers where the number of outliers is set to be $z'=(1+\epsilon)\gamma |S|$. If $E$ is an $\alpha$-approximate solution of $(k, z')_{\epsilon}$-center clustering on $S$,  then $E$ is an $\alpha$-approximate solution of $(k, z)_{O(\epsilon)}$-center clustering on $P$, with constant probability.
\end{theorem}
\begin{proof}
We assume that $S$ is a relative approximation of $P$ and (\ref{for-relativesample}) holds (this happens with constant probability). Let $\mathbb{B}_{opt}$ be the set of $k$ balls covering $(1-\gamma)n$ points induced by the optimal solution for $P$, and $\mathbb{B}_{S}$ be the set of $k$ balls induced by an $\alpha$-approximate solution of $(k, z')_{\epsilon}$-center clustering on $S$. Suppose the radius of each ball in $\mathbb{B}_{opt}$ (resp., $\mathbb{B}_{S}$) is $r_{opt}$ (resp., $r_S$). We denote the complements of $\mathbb{B}_{opt}$ and $\mathbb{B}_{S}$ as $\pi_{opt}$ and $\pi_{S}$, respectively.

First, since $\mathbb{B}_{opt}$ covers $(1-\gamma)n$ points of $P$ and $S$ is a relative approximation of $P$, we have
\begin{eqnarray}
\frac{\big|\pi_{opt}\cap S\big| }{|S|}\leq  \frac{\big|\pi_{opt}\cap P\big| }{|P|}+ \epsilon\times\max\Big\{\frac{|\pi_{opt}\cap P|}{|P|}, \gamma\Big\}= (1+\epsilon)\gamma \label{for-samplereduce4}
\end{eqnarray}
by (\ref{for-relativesample}). That is, the set balls $\mathbb{B}_{opt}$ cover at least $\big(1-(1+\epsilon)\gamma\big) |S|$ points of $S$, and therefore it is a feasible solution for the instance $S$ with respect to the problem of $k$-center clustering with $z'$ outliers. 
%
%
%
Since $\mathbb{B}_{S}$ is  an $\alpha$-approximate solution of $(k, z')_{\epsilon}$-center clustering on $S$, we have
\begin{eqnarray}
r_S \leq  \alpha r_{\textrm{opt}};  \hspace{0.2in}
|\pi_{S}\cap S| \leq (1+\epsilon)z'=(1+\epsilon)^2\gamma|S|. \label{for-samplereduce2}
\end{eqnarray}

Now, we claim that
\vspace{-0.2in}
\begin{eqnarray}
\big|\pi_{S}\cap P\big|\leq \frac{(1+\epsilon)^2}{1-\epsilon}\gamma |P|. \label{for-samplereduce3}
\end{eqnarray}
Assume that (\ref{for-samplereduce3}) is not true, then (\ref{for-relativesample}) implies
$\Big|\frac{|\pi_{S}\cap P|}{|P|}-\frac{|\pi_{S}\cap S|}{|S|}\Big|\leq \epsilon\times \max\Big\{\frac{|\pi_{S}\cap P|}{|P|},\gamma\Big\}=\epsilon \frac{|\pi_{S}\cap P|}{|P|}$.
So $\frac{|\pi_{S}\cap S|}{|S|}\geq (1-\epsilon)\frac{|\pi_{S}\cap P|}{|P|}>(1+\epsilon)^2\gamma$, which is in contradiction with the second inequality of (\ref{for-samplereduce2}), and thus (\ref{for-samplereduce3}) is true. We assume $\epsilon\leq 0.5$, so $\frac{1}{1-\epsilon}\leq 1+2\epsilon$ and $\frac{(1+\epsilon)^2}{1-\epsilon}=1+O(\epsilon)$. Consequently (\ref{for-samplereduce3}) and the first inequality of (\ref{for-samplereduce2}) together imply that $\mathbb{B}_{S}$ is an $\alpha$-approximate solution of $(k, z)_{O(\epsilon)}$-center clustering on $P$.
%
\qed\end{proof}

\vspace{-0.15in}
\section{Coreset Construction in Doubling Metrics}
\label{sec-doubling}
\vspace{-0.1in}

%
%
%
In this section, we always assume the following is true by default:

\vspace{0.05in}
{\em Given an instance $(X,d)$ of $k$-center clustering with outliers, the metric $(X_{opt},d)$, i.e., the metric formed by the set of inliers, has a constant doubling dimension $\rho>0$.}
 \vspace{0.05in}
 
We do not have any restriction on the outliers $X\setminus X_{opt}$. Thus the above assumption is more relaxed and practical than assuming the whole $(X,d)$ has a constant doubling dimension. 
%
From Definition~\ref{def-dd}, we directly know that each optimal cluster $C_l$ of $X_{opt}$ can be covered by $2^\rho$ balls with radius $r_{opt}/2$ (see the left figure in Figure.~\ref{fig-dd}). Imagine that the instance $(X, d)$ has $2^\rho k$ clusters, where the optimal radius is at most $r_{opt}/2$. Therefore, we can just replace $k$ by $2^\rho k$ when running Algorithm~\ref{alg-bi}, so as to reduce the approximation ratio (i.e., the ratio of the resulting radius to $r_{opt}$) from $2$ to $1$. 

\begin{wrapfigure}{r}{0.45\textwidth}
  \vspace{-30pt}
\begin{center}
    \includegraphics[width=0.37\textwidth]{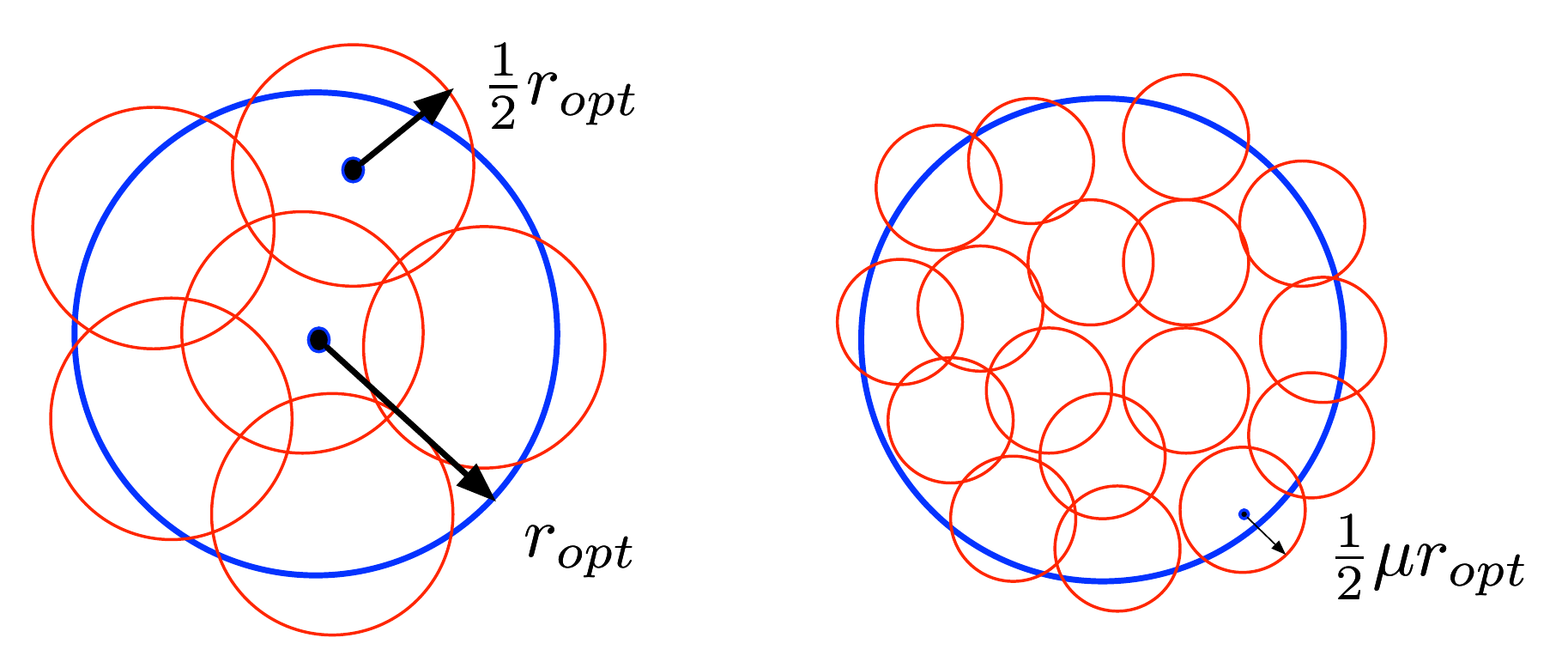}  
    \end{center}
  \vspace{-22pt}
  \caption{Illustrations for Theorem~\ref{the-double-biapprox} and \ref{the-coreset}.}   
   \label{fig-dd}
     \vspace{-18pt}
\end{wrapfigure}

\vspace{-0.1in}
\begin{theorem}
\label{the-double-biapprox}
If we set $t=\frac{2^\rho k+2^{\rho/2}\sqrt{k}}{1-\eta}$ for Algorithm~\ref{alg-bi}, with probability $(1-\eta)(1-e^{-\frac{1-\eta}{4}})$, $\phi_{\epsilon}(X,E)\leq  r_{opt}$. So the set $E$ is a $\big(1, O(\frac{2^\rho}{\epsilon})\big)$-approximation for the problem of $(k,z)_{\epsilon}$-center clustering, and the running time is $O(2^\rho\frac{k}{\epsilon}n)$.
\end{theorem}
\vspace{-0.1in}

If considering the problem in Euclidean space $\mathbb{R}^D$ where the doubling dimension of the inliers is $\rho$, the running time becomes $O(2^\rho\frac{k}{\epsilon}nD)$. 
Inspired by Theorem~\ref{the-double-biapprox}, we can further construct coreset for $k$-center clustering with outliers (see Definition~\ref{def-coreset}). Let $\mu\in (0,1)$. If applying Definition~\ref{def-dd} recursively, we know that each $C_l$ is covered by $2^{\rho\log 2/\mu}=(\frac{2}{\mu})^\rho$ balls with radius $\frac{\mu}{2} r_{opt}$, and $X_{opt}$ is covered by $(\frac{2}{\mu})^\rho k$ such balls in total. See the right figure in Figure.~\ref{fig-dd}. Based on this observation, we have Algorithm~\ref{alg-coreset} for constructing $\mu$-coreset.

\vspace{-0.05in}
\begin{theorem}
\label{the-coreset}
With constant probability, Algorithm~\ref{alg-coreset} outputs a $\mu$-coreset $E$ of $k$-center clustering with $z$ outliers. The size of $E$ is at most $2z+O\big((\frac{2}{\mu})^\rho k\big)$, and the construction time is $O((\frac{2}{\mu})^\rho kn)$.
\end{theorem}
\textbf{Remark.} \textbf{(1)} The previous ideas based on uniform sampling~\cite{charikar2003better,huang2018epsilon} (also our idea in Section~\ref{sec-core1}) cannot get rid of the violation on the number of outliers; the sample sizes will become infinity if not allowing to remove more than $z$ outliers. Our coreset in Theorem~\ref{the-coreset} works for removing $z$ outliers exactly. Consequently, our coreset can be used for existing algorithms of $k$-center clustering with outliers, such as \cite{charikar2001algorithms}, to reduce their complexities. \textbf{(2)} Another feature is that our coreset is a natural composable coreset. If $X$ (or the point set $P$) is partitioned into $L$ parts, we can run Algorithm~\ref{alg-coreset} for each part, and obtain a coreset with size $\Big(2z+O\big((\frac{2}{\mu})^\rho k\big)\Big) L$ in total (the proof is almost identical to the proof of Theorem~\ref{the-coreset} below). So our coreset construction can potentially be applied to distributed clustering with outliers. \textbf{(3)} The coefficient $2$ of $z$ actually can be further reduced by modifying the value of $\epsilon$ in Step 2 of Algorithm~\ref{alg-coreset} (we just set $\epsilon=1$ for simplicity). In general, the size of $E$ is $(1+\epsilon)z+O\big(\frac{1}{\epsilon}(\frac{2}{\mu})^\rho k\big)$ and  the construction time is $O(\frac{1}{\epsilon}(\frac{2}{\mu})^\rho kn)$ (or $O(\frac{1}{\epsilon}(\frac{2}{\mu})^\rho knD)$ in $\mathbb{R}^D$).
%
%
\begin{algorithm}[tb]
   \caption{The Coreset Construction}
   \label{alg-coreset}
\begin{algorithmic}
  \STATE {\bfseries Input:} An instance $(X,d)$ of metric $k$-center clustering with $z$ outliers, and $|X|=n$; parameters $\eta$ and $\mu\in (0,1)$.
   \STATE
\begin{enumerate}
\item Let $l=(\frac{2}{\mu})^\rho k$.


\item Set $\epsilon=1$ and run Algorithm~\ref{alg-bi} $t=\frac{l+\sqrt{l}}{1-\eta}$ rounds.
Denote by $\tilde{r}$ the maximum distance between $E$ and $X$ by excluding the farthest $2z$ vertices, after the final round of Algorithm~\ref{alg-bi}.

\item Let 
$X_{\tilde{r}}=\{p\mid p\in X \text{ and } d(p, E)\leq \tilde{r}\}$.

\item For each vertex $p\in X_{\tilde{r}}$, assign it to its nearest neighbor in $E$; for each vertex $q\in E$, let its weight be the number of vertices assigning to it.


\item Add $X\setminus X_{\tilde{r}}$ to $E$;  each vertex of $X\setminus X_{\tilde{r}}$ has weight $1$.

\end{enumerate}
  \STATE {\bfseries Output} $E$ as the coreset.
\end{algorithmic}
\end{algorithm}

\begin{proof}[Proof of Theorem~\ref{the-coreset}]
Similar to Theorem~\ref{the-double-biapprox}, we know that  $|X_{\tilde{r}}|= n-2z$ and $\tilde{r}\leq 2\times \frac{\mu}{2} r_{opt}=\mu r_{opt}$ with constant probability in Algorithm~\ref{alg-coreset}. Thus, the size of $E$ is $|X\setminus X_{\tilde{r}}|+O\big((\frac{2}{\mu})^\rho k\big)= 2z+O\big((\frac{2}{\mu})^\rho k\big)$. Moreover, it is easy to see that the running time of Algorithm~\ref{alg-coreset} is $O((\frac{2}{\mu})^\rho kn)$.


Next, we show that $E$ is a $\mu$-coreset of $X$. For each vertex $q\in E$, denote by $w(q)$ the weight of $q$; for the sake of convenience in our proof, we view each $q$ as a set of $w(q)$ overlapping unit weight vertices. Thus, from the construction of $E$, we can see that there is a bijective mapping $f$ between $X$ and $E$, where 
\begin{eqnarray}
||p-f(p)||\leq\tilde{r}\leq\mu r_{opt}, \hspace{0.2in} \forall p\in X. \label{for-map}
\end{eqnarray}
Let $H=\{c_1, c_2, \cdots, c_k\}$ be any $k$ vertices of $X$. Suppose that $H$ induces $k$ clusters $\{A_1, A_2, \cdots, A_k\}$ (resp., $\{B_1, B_2, \cdots, B_k\}$) with respect to the problem of $k$-center clustering with $z$ outliers on $E$ (resp., $X$), where each $A_j$ (resp., $B_j$) has the cluster center $c_j$ for $1\leq j\leq k$. Let $r_E=\phi_0(E, H)$ and $r_X=\phi_0(X, H)$, respectively. Also, let $r'_E$ (resp., $r'_X$) be the smallest value $r$,  such that for any $1\leq j\leq k$, $f(B_j)\subseteq Ball(c_j, r)$ (resp., $f^{-1}(A_j)\subseteq Ball(c_j, r)$). We need the following claim.

\begin{claim}
\label{cla-core}
$|r'_E-r_X|\leq \mu r_{opt}$ and $|r'_X-r_E|\leq \mu r_{opt}$ (see the proof in Section~\ref{sec-proof-cla-core}).
\end{claim}
In addition, since $\{f(B_1), \cdots, f(B_k)\}$ also form $k$ clusters for the instance $E$ with the fixed $k$ cluster centers of $H$, we know that
$r'_E\geq \phi_0(E, H)=r_E$. 
Similarly, we have 
$r'_X\geq r_X$.  
Combining Claim~\ref{cla-core}, we have 
\begin{eqnarray}
r_X-\mu r_{opt}\leq \underbrace{r'_X-\mu r_{opt}\leq r_E}_{\text{by Claim~\ref{cla-core}}}\leq \underbrace{r'_E\leq r_X+\mu r_{opt}}_{\text{by Claim~\ref{cla-core}}}.
\end{eqnarray}
So $|r_X-r_E|\leq \mu r_{opt}$, i.e., $\phi_0(E, H)\in \phi_0(X, H)\pm \mu r_{opt}\subseteq(1\pm \mu) \phi_0(X, H)$.  Therefore $E$ is a $\mu$-coreset of $(X,d)$.
%
%
%
%
%
%
%
%
%
\qed\end{proof}

\section{Experiments}
\label{sec-exp}
\vspace{-0.1in}
Our experimental results were obtained on a Windows workstation with 2.8GHz Intel(R) Core(TM) i5-840 and 8GB main memory; the algorithms were implemented in Matlab R2018a. We test our algorithms on both synthetic and real datasets. For Algorithm~\ref{alg-single}, we take two well known algorithms of $k$-center clustering with outliers,  $Base1$ of~\cite{charikar2001algorithms} and $Base2$ of~\cite{mccutchen2008streaming}, as the baselines. For Algorithm~\ref{alg-coreset}, we compare our coreset construction with uniform random sampling. 

To generate the synthetic datasets, we set $n=10^{5}$ and $D=10^{3}$, and vary the values of $z$ and $k$. 
First, randomly generate $k$ clusters inside a hypercube of side length $200$, where each cluster is a random sample from a Gaussian distribution with variance $10$; each cluster has a random number of points and we keep the total number of points to be $n-z$; we compute the minimum enclosing balls respectively for these $k$ clusters (by using the algorithm of \cite{badoiu2003smaller}), and randomly generate $z$ outliers outside the balls. The maximum radius of the balls is used as $r_{opt}$.

We also use three real datasets. MNIST dataset~\cite{lecun1998gradient} contains $n=60,000$ handwritten digit images from $0$ to $9$, where each image is represented by a $784$-dimensional vector. The $10$ digits form $k=10$ clusters.  
Caltech-256 dataset~\cite{fei2007learning} contains $30,607$ colored images with 256 categories, where each image is represented by a $4096$-dimensional vector. We choose $n=2,232$ images of 20 categories to form $k=20$ clusters.
CIFER-10 training dataset~\cite{krizhevsky2009learning} contains $n=50,000$ colored images in 10 classes as $k=10$ clusters, where each image is represented by a $4096$-dimensional vector. For each real dataset, we use the minimum enclosing ball algorithm of \cite{badoiu2003smaller} to compute $r_{opt}$, and randomly generate $z=5\% n$ outliers outside the corresponding balls.

\textbf{Results and analysis.} Note that we exactly exclude $z$ outliers (rather than $(1+\epsilon)z$ as stated in Theorem~\ref{the-biapprox} and \ref{the-kcenter}) in our experiments, and calculate the approximation ratio $\phi_0(X, E)/r_{opt}$ for each instance, if $E$ is the set of returned cluster centers. 

We first run our Algorithm~\ref{alg-bi} on synthetic and real datasets. For synthetic datasets, we set $k=2$-$20$, and $\beta=|E|/k=8$  via modifying the values of $\epsilon$ and $\eta$ appropriately (that means we output $8k$ cluster centers); normally, we set $\eta=0.1$ and $\epsilon\approx 0.7$. We try the instances with $z=\{2\%n, 4\%n, 6\%n, 8\%n, 10\%n\}$, and report the average results in Figure~\ref{fig-a1kr} and \ref{fig-a1kt}; the approximation ratios are within $1.3$-$1.4$ and the running times are less than $30$s. Actually, the performance is quite stable regarding different values of $z$ in our experiments, and the standard variances of approximation ratios and running times are less than $0.03$ and $0.12$, respectively. We also vary the value of $\beta$ from $4$ to $28$ with $k=10$. Figure~\ref{fig-a1er} shows that the approximation ratio slightly decreases as $\beta$ increases. The running times are all around $14$s and do not reveal a clear increasing trend as $\beta$ increases. We think the reason behind may be that we just use the simple $O(n\log n)$ sorting algorithm, rather than the linear time selection algorithm~\cite{blum1973time}, for computing $Q_j$ in practice (see Step 3(a) of Algorithm~\ref{alg-bi}); thus the running time is not linearly dependent on $|E|$. The results for real datasets are shown in Section~\ref{sec-exp-alg1}; the approximation ratios are all below $1.3$ and the running times are less than $35$s even for the largest CIFER-10 dataset.
\begin{figure}[h]
\vspace{-.6cm} 
\centering
  \subfloat[]{\label{fig-a1kr}\includegraphics[height=1in]{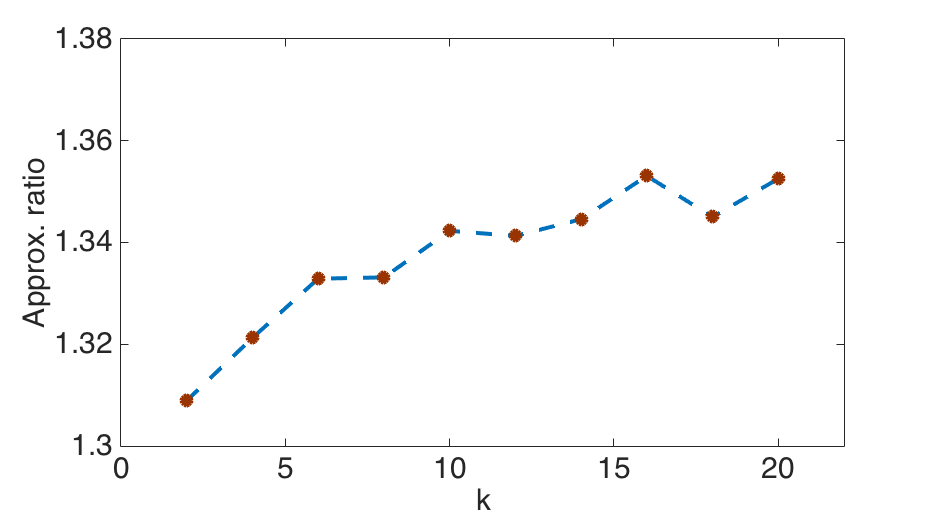}}
  \subfloat[]{\label{fig-a1kt}\includegraphics[height=1in]{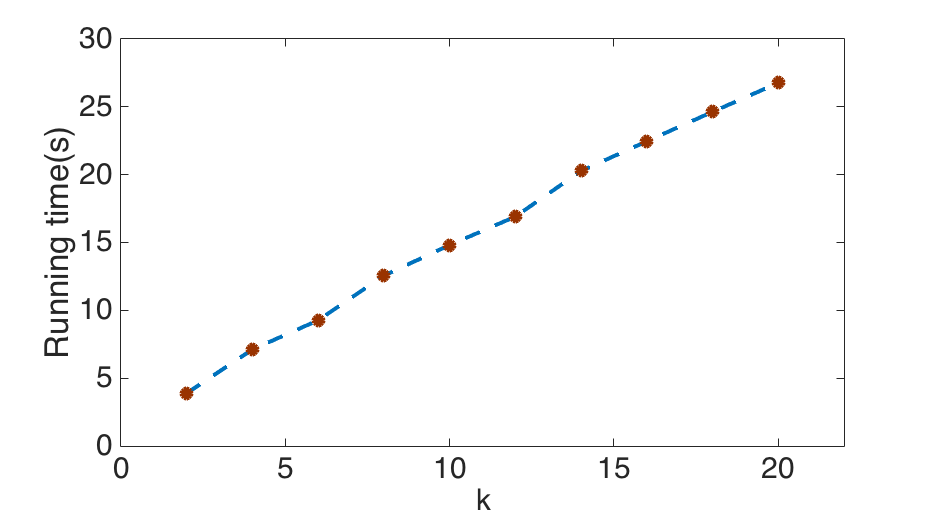}}
    \subfloat[]{\label{fig-a1er}\includegraphics[height=1in]{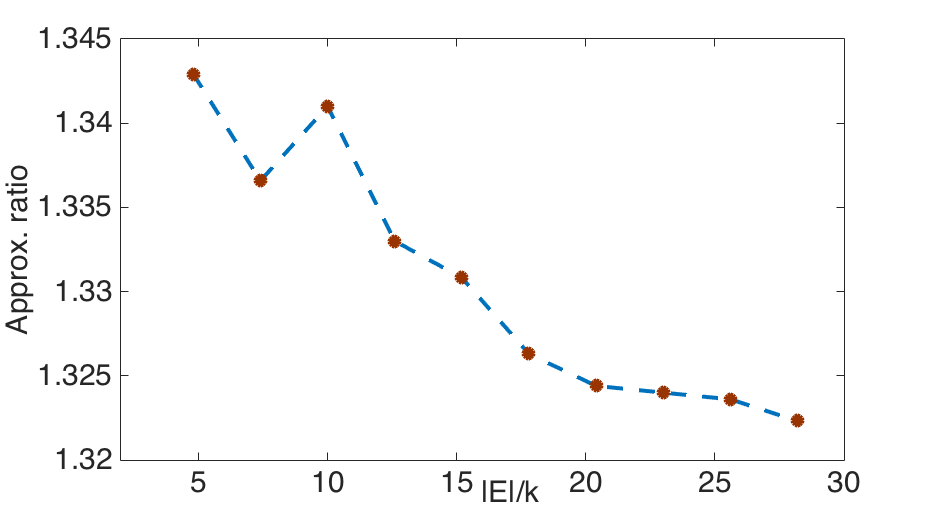}}
         \vspace{-0.1in}
  \caption{The experimental results of Algorithm~\ref{alg-bi} on synthetic datasets.} 
  \vspace{-0.2in}
\end{figure}

We also test our Algorithm~\ref{alg-single} on synthetic and real datasets. We set $\epsilon=1$ so that to avoid to repeat running  Algorithm~\ref{alg-single} too many times (see Corollary~\ref{the-kcenter2}), but we still exactly exclude $z$ outliers for calculating the approximation ratio as mentioned before. Our results are shown in Table~\ref{tab-alg2}. The synthetic and real datasets are too large to the baseline algorithms $Base1$ and $Base2$, e.g., they run too slowly or even out of memory in our workstation if $n$, $z$, and $D$ are large (they have complexities $\Omega(n^2 D)$ or $\Omega(kznD)$)\footnote{We are aware of several distributed algorithms for $k$-center clustering with outliers\cite{malkomes2015fast,guha2017distributed,DBLP:journals/corr/abs-1802-09205,li2018distributed}, but we only consider the setting with single machine in this paper.}. To make a fair comparison, we run $Base1$, $Base2$, and Algorithm~\ref{alg-single} on smaller synthetic datasets with $(n=2000, D=10)$ and $(n=2000, D=100)$; we also set $z=\{2\%n, 4\%n, 6\%n, 8\%n, 10\%n\}$ as before and report the average results. When $D=10$, $Base1$ and Algorithm~\ref{alg-single} achieve approximation ratios $<1.5$ generally (Figure~\ref{fig-a210r}); moreover, $Base2$ and Algorithm~\ref{alg-single} run much faster than $Base1$ (Figure~\ref{fig-a210t}). However, when $D=100$, $Base1$ and $Base2$ yield much worse approximation ratios than Algorithm~\ref{alg-single} (Figure~\ref{fig-a2100r} and \ref{fig-a2100t}). Our experiment reveals that Algorithm~\ref{alg-single} can achieve a more stable performance when dimensionality increases.

\begin{figure}[h]
\vspace{-.8cm} 
\centering
     \subfloat[]{\label{fig-a210r}\includegraphics[height=0.89in]{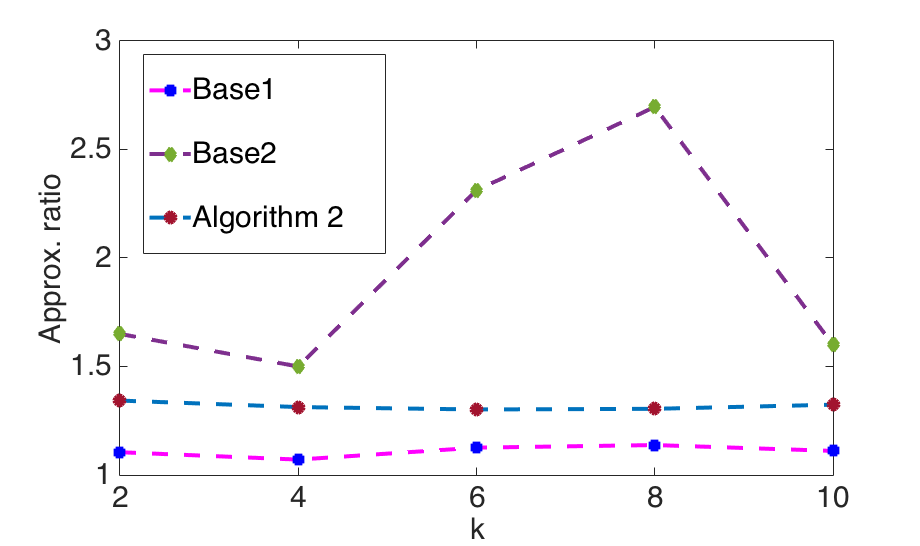}}
  \subfloat[]{\label{fig-a210t}\includegraphics[height=0.89in]{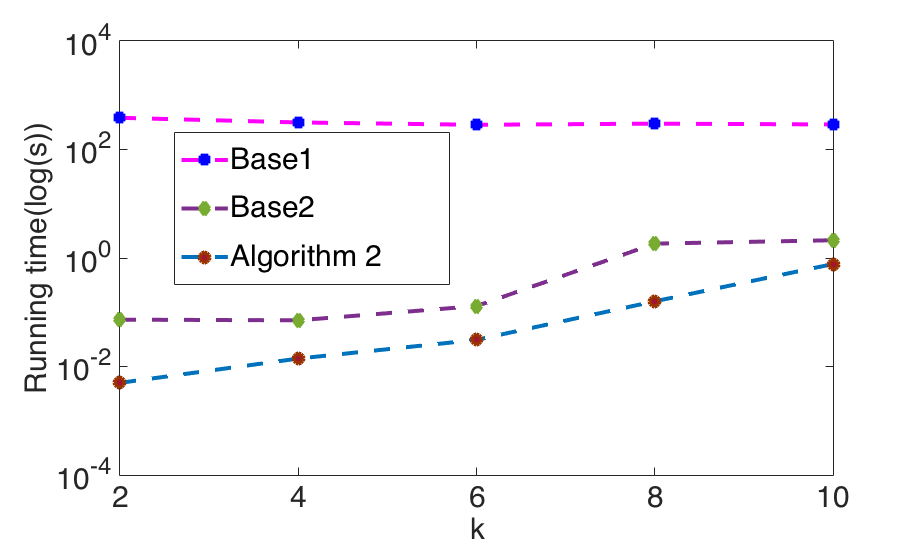}}
     \subfloat[]{\label{fig-a2100r}\includegraphics[height=0.89in]{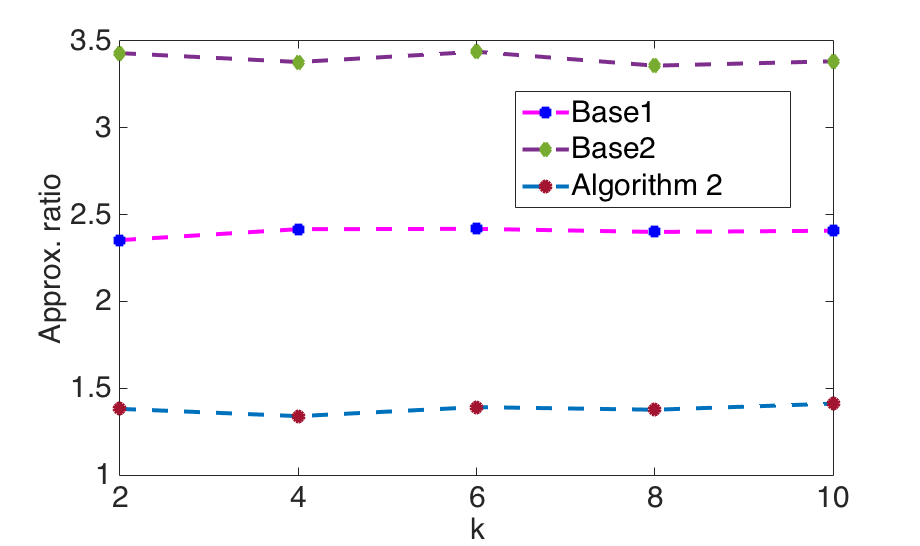}}
  \subfloat[]{\label{fig-a2100t}\includegraphics[height=0.89in]{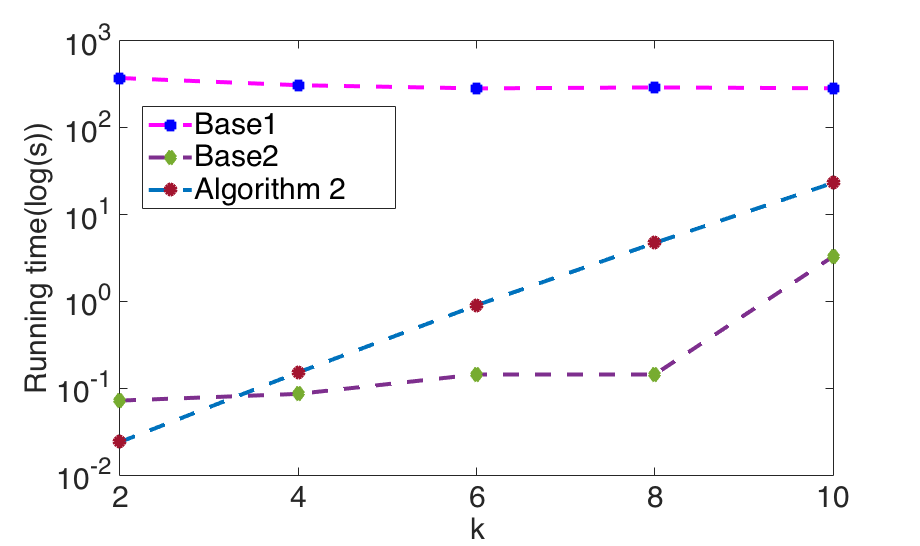}} 
      \vspace{-0.1in}
  \caption{Comparison of Base1, Base2, and Algorithm~\ref{alg-single} on smaller synthetic datasets ((a) and (b) for $D=10$; (c) and (d) for $D=100$).} 
  \vspace{-0.25in}
\end{figure}

\begin{table}[h]
\vspace{-0.35in}
	\centering
	\caption{The results of Algorithm~\ref{alg-single} on synthetic and real datasets}
	\begin{tabular}{|c|c|c|c|c|c|c|c|}
		\hline
		\multirow{2}[4]{*}{} & \multicolumn{4}{c|}{Synthetic datasets} & \multicolumn{3}{c|}{Real datasets} \\
		\cline{2-8}          & k=2   & k=4   & k=6   & k=8   & MNIST & CALTECH256 & CIFAR10  \\
		\hline
		Approx. ratio& 1.410 & 1.403 & 1.406 & 1.423 & 1.277 & 1.368 & 1.249  \\
		\hline
		Running time(s) & 8.097 & 63.636 & 374.057 & 1939.004 & 2644.709 & 3381.421 & 13295.306 \\
		\hline
	\end{tabular}%
	\label{tab-alg2}%
	\vspace{-0.2in}
\end{table}%

Finally, we compare the performances of our coresets method (Algorithm~\ref{alg-coreset}) and uniform random sampling in terms of reducing data sizes. Though real-world image datasets often are believed to have low intrinsic dimenions~\cite{belkin2003problems}, it is difficult to compute them (e.g., doubling dimension) accurately. In practice, we can directly set an appropriate value for $l$ in Step 1 of Algorithm~\ref{alg-coreset} (without knowing the value of doubling dimension $\rho$). For example, the size of coreset is $2z+O\big((\frac{2}{\mu})^\rho k\big)=2z+O(l)$ according to Theorem~\ref{the-coreset}, so we keep the sizes of our coresets to be $\{15\%n, 20\%n, 25\%n\}$ via modifying the value of $l$ in our experiments. Correspondingly, we also set the sizes of random samples to be  $\{15\%n, 20\%n, 25\%n\}$. We run Algorithm~\ref{alg-single} on the corresponding random samples and coresets, and report the results in Table~\ref{tab-core}. Running Algorithm~\ref{alg-single} on the coresets yields approximation ratios close to those obtained by directly running the algorithm on the original datasets; the results also remain stably when the level reduces from $25\%$ to $15\%$. More importantly, our coresets significantly reduce the running times (e.g., it only needs $15\%$-$35\%$ time by using $15\%$-level coreset). Comparing with the random samples, our coresets can achieve significantly lower approximation ratios especially for the $15\%$ level. Note that the coreset based approach takes more time than uniform random sampling, because we count the time spent for coreset construction.


\begin{table}[h]
	\vspace{-0.3in}
	\centering
	\caption{The results of  Algorithm~\ref{alg-single} on random samples, coresets, and original datasets}
	\begin{tabular}{|c|c|c|c|c|c|c|c|c|}
		\hline
		\multicolumn{2}{|c|}{\multirow{2}[4]{*}{}} & \multicolumn{3}{c|}{random sampling} & \multicolumn{3}{c|}{coreset} & \multirow{2}[4]{*}{100\%} \\
		\cline{3-8}    \multicolumn{2}{|c|}{} & 15\%  & 20\%  & 25\%  & 15\%  & 20\%  & 25\%  &  \\
		\hline
		MNIST & Appro. Ratio & 1.591 & 1.597 & 1.566 & 1.275 & 1.261 & 1.261 & 1.277 \\
		\cline{2-9}          & running time(s) & 624.612 & 769.517 & 958.549 & 936.393 & 1071.926 & 1262.996 & 2644.709 \\
		\hline
		CALTECH256 & Appro. Ratio & 2.903 & 2.826 & 1.935 & 1.713 & 1.722 & 1.701 & 1.368 \\
		\cline{2-9}          & running time(s) & 486.647 & 502.605 & 598.059 & 489.625 & 505.537 & 600.900 & 3381.421 \\
		\hline
		CIFAR10 & Appro. Ratio & 1.538 & 1.383 & 1.446 & 1.248 & 1.256 & 1.249 & 1.249 \\
		\cline{2-9}          & running time(s) & 2420.943 & 2170.416 & 2938.773 & 3526.752 & 3264.858 & 4033.862 & 13295.306 \\
		\hline
	\end{tabular}%
	\label{tab-core}%
		\vspace{-0.3in}
\end{table}%

\vspace{-0.05in}

\section{Future Work}
\vspace{-0.1in}
 Following our work, several interesting problems deserve to be studied in future. For example, 
can the coreset construction time of Algorithm~\ref{alg-coreset} be improved, like the fast net construction method proposed by Har-Peled and Mendel~\cite{har2006fast} in doubling metrics? It is also interesting to study other problems involving outliers by using greedy strategy.

\newpage
\bibliographystyle{abbrv}

\bibliography{greedy}

\begin{thebibliography}{10}

\bibitem{aggarwal2009adaptive}
A.~Aggarwal, A.~Deshpande, and R.~Kannan.
\newblock Adaptive sampling for k-means clustering.
\newblock In {\em Approximation, Randomization, and Combinatorial Optimization.
  Algorithms and Techniques}, pages 15--28. Springer, 2009.

\bibitem{DBLP:conf/cccg/AghamolaeiG18}
S.~Aghamolaei and M.~Ghodsi.
\newblock A composable coreset for k-center in doubling metrics.
\newblock In {\em Proceedings of the 30th Canadian Conference on Computational
  Geometry, {CCCG} 2018, August 8-10, 2018, University of Manitoba, Winnipeg,
  Manitoba, Canada}, pages 165--171, 2018.

\bibitem{alon2003testing}
N.~Alon, S.~Dar, M.~Parnas, and D.~Ron.
\newblock Testing of clustering.
\newblock {\em SIAM Journal on Discrete Mathematics}, 16(3):393--417, 2003.

\bibitem{alon2004probabilistic}
N.~Alon and J.~H. Spencer.
\newblock {\em The probabilistic method}.
\newblock John Wiley \& Sons, 2004.

\bibitem{awasthi2014center}
P.~Awasthi and M.-F. Balcan.
\newblock Center based clustering: A foundational perspective.
\newblock 2014.

\bibitem{bachem2017practical}
O.~Bachem, M.~Lucic, and A.~Krause.
\newblock Practical coreset constructions for machine learning.
\newblock {\em arXiv preprint arXiv:1703.06476}, 2017.

\bibitem{badoiu2003smaller}
M.~Badoiu and K.~L. Clarkson.
\newblock Smaller core-sets for balls.
\newblock In {\em Proceedings of the ACM-SIAM Symposium on Discrete Algorithms
  (SODA)}, pages 801--802, 2003.

\bibitem{BHI}
M.~Badoiu, S.~Har-Peled, and P.~Indyk.
\newblock Approximate clustering via core-sets.
\newblock In {\em Proceedings of the ACM Symposium on Theory of Computing
  (STOC)}, pages 250--257, 2002.

\bibitem{belkin2003problems}
M.~Belkin.
\newblock {\em Problems of learning on manifolds}.
\newblock The University of Chicago, 2003.

\bibitem{blum1973time}
M.~Blum, R.~W. Floyd, V.~Pratt, R.~L. Rivest, and R.~E. Tarjan.
\newblock Time bounds for selection.
\newblock {\em Journal of Computer and System Sciences}, 7(4):448--461, 1973.

\bibitem{blumer1989learnability}
A.~Blumer, A.~Ehrenfeucht, D.~Haussler, and M.~K. Warmuth.
\newblock Learnability and the vapnik-chervonenkis dimension.
\newblock {\em Journal of the ACM (JACM)}, 36(4):929--965, 1989.

\bibitem{DBLP:journals/corr/abs-1802-09205}
M.~Ceccarello, A.~Pietracaprina, and G.~Pucci.
\newblock Solving k-center clustering (with outliers) in mapreduce and
  streaming, almost as accurately as sequentially.
\newblock {\em CoRR}, abs/1802.09205, 2018.

\bibitem{DBLP:conf/icalp/ChakrabartyGK16}
D.~Chakrabarty, P.~Goyal, and R.~Krishnaswamy.
\newblock The non-uniform k-center problem.
\newblock In {\em 43rd International Colloquium on Automata, Languages, and
  Programming, {ICALP} 2016, July 11-15, 2016, Rome, Italy}, pages 67:1--67:15,
  2016.

\bibitem{chandola2009anomaly}
V.~Chandola, A.~Banerjee, and V.~Kumar.
\newblock Anomaly detection: A survey.
\newblock {\em ACM Computing Surveys (CSUR)}, 41(3):15, 2009.

\bibitem{charikar2001algorithms}
M.~Charikar, S.~Khuller, D.~M. Mount, and G.~Narasimhan.
\newblock Algorithms for facility location problems with outliers.
\newblock In {\em Proceedings of the twelfth annual ACM-SIAM symposium on
  Discrete algorithms}, pages 642--651. Society for Industrial and Applied
  Mathematics, 2001.

\bibitem{charikar2003better}
M.~Charikar, L.~O'Callaghan, and R.~Panigrahy.
\newblock Better streaming algorithms for clustering problems.
\newblock In {\em Proceedings of the thirty-fifth annual ACM symposium on
  Theory of computing}, pages 30--39. ACM, 2003.

\bibitem{DBLP:journals/corr/abs-1205-4891}
A.~Daniely, N.~Linial, and M.~E. Saks.
\newblock Clustering is difficult only when it does not matter.
\newblock {\em CoRR}, abs/1205.4891, 2012.

\bibitem{gonzalez1985clustering}
T.~F. Gonzalez.
\newblock Clustering to minimize the maximum intercluster distance.
\newblock {\em Theoretical Computer Science}, 38:293--306, 1985.

\bibitem{guha2017distributed}
S.~Guha, Y.~Li, and Q.~Zhang.
\newblock Distributed partial clustering.
\newblock In {\em Proceedings of the 29th ACM Symposium on Parallelism in
  Algorithms and Architectures}, pages 143--152. ACM, 2017.

\bibitem{har2006fast}
S.~Har-Peled and M.~Mendel.
\newblock Fast construction of nets in low-dimensional metrics and their
  applications.
\newblock {\em SIAM Journal on Computing}, 35(5):1148--1184, 2006.

\bibitem{har2011relative}
S.~Har-Peled and M.~Sharir.
\newblock Relative (p, $\varepsilon$)-approximations in geometry.
\newblock {\em Discrete \& Computational Geometry}, 45(3):462--496, 2011.

\bibitem{hochbaum1985best}
D.~S. Hochbaum and D.~B. Shmoys.
\newblock A best possible heuristic for the k-center problem.
\newblock {\em Mathematics of operations research}, 10(2):180--184, 1985.

\bibitem{huang2018epsilon}
L.~Huang, S.~Jiang, J.~Li, and X.~Wu.
\newblock Epsilon-coresets for clustering (with outliers) in doubling metrics.
\newblock In {\em 2018 IEEE 59th Annual Symposium on Foundations of Computer
  Science (FOCS)}, pages 814--825. IEEE, 2018.

\bibitem{DBLP:conf/pods/IndykMMM14}
P.~Indyk, S.~Mahabadi, M.~Mahdian, and V.~S. Mirrokni.
\newblock Composable core-sets for diversity and coverage maximization.
\newblock In {\em Proceedings of the 33rd {ACM} {SIGMOD-SIGACT-SIGART}
  Symposium on Principles of Database Systems, PODS'14, Snowbird, UT, USA, June
  22-27, 2014}, pages 100--108, 2014.

\bibitem{jain2010data}
A.~K. Jain.
\newblock Data clustering: 50 years beyond k-means.
\newblock {\em Pattern recognition letters}, 31(8):651--666, 2010.

\bibitem{DBLP:journals/pami/KanungoMNPSW02}
T.~Kanungo, D.~M. Mount, N.~S. Netanyahu, C.~D. Piatko, R.~Silverman, and A.~Y.
  Wu.
\newblock An efficient k-means clustering algorithm: Analysis and
  implementation.
\newblock {\em {IEEE} Trans. Pattern Anal. Mach. Intell.}, 24(7):881--892,
  2002.

\bibitem{krizhevsky2009learning}
A.~Krizhevsky.
\newblock Learning multiple layers of features from tiny images.
\newblock Technical report, Citeseer, 2009.

\bibitem{lecun1998gradient}
Y.~LeCun, L.~Bottou, Y.~Bengio, and P.~Haffner.
\newblock Gradient-based learning applied to document recognition.
\newblock {\em Proceedings of the IEEE}, 86(11):2278--2324, 1998.

\bibitem{fei2007learning}
F.-F. Li, R.~Fergus, and P.~Perona.
\newblock Learning generative visual models from few training examples: An
  incremental bayesian approach tested on 101 object categories.
\newblock {\em Computer vision and Image understanding}, 106(1):59--70, 2007.

\bibitem{li2018distributed}
S.~Li and X.~Guo.
\newblock Distributed $ k $-clustering for data with heavy noise.
\newblock In {\em Advances in Neural Information Processing Systems}, pages
  7849--7857, 2018.

\bibitem{li2001improved}
Y.~Li, P.~M. Long, and A.~Srinivasan.
\newblock Improved bounds on the sample complexity of learning.
\newblock {\em Journal of Computer and System Sciences}, 62(3):516--527, 2001.

\bibitem{malkomes2015fast}
G.~Malkomes, M.~J. Kusner, W.~Chen, K.~Q. Weinberger, and B.~Moseley.
\newblock Fast distributed k-center clustering with outliers on massive data.
\newblock In {\em Advances in Neural Information Processing Systems}, pages
  1063--1071, 2015.

\bibitem{mccutchen2008streaming}
R.~M. McCutchen and S.~Khuller.
\newblock Streaming algorithms for k-center clustering with outliers and with
  anonymity.
\newblock In {\em Approximation, Randomization and Combinatorial Optimization.
  Algorithms and Techniques}, pages 165--178. Springer, 2008.

\bibitem{DBLP:journals/corr/Phillips16}
J.~M. Phillips.
\newblock Coresets and sketches.
\newblock {\em Computing Research Repository}, 2016.

\bibitem{talwar2004bypassing}
K.~Talwar.
\newblock Bypassing the embedding: algorithms for low dimensional metrics.
\newblock In {\em Proceedings of the thirty-sixth annual ACM symposium on
  Theory of computing}, pages 281--290, 2004.

\bibitem{tan2006introduction}
P.-N. Tan, M.~Steinbach, and V.~Kumar.
\newblock {\em Introduction to Data Mining}.
\newblock 2006.

\end{thebibliography}

\section{Proof of Claim \ref{pro-core}}
\label{sec-proof-c1}

Suppose $H$ is an $\alpha$-approximation of the instance (coreset) $S$. Let $H_{opt}$ be the set of $k$ cluster centers yielding the optimal solution of $X$. Then we have 
\begin{eqnarray}
\phi_{0}(S, H)&\leq&\alpha \phi_{0}(S, H_{opt});  \label{for-pro13}\\
\phi_0 (S,H)&\in& (1\pm\mu)\phi_{0}(X, H); \label{for-pro11}\\
\phi_0 (S,H_{opt})&\in& (1\pm\mu)\phi_{0}(X, H_{opt}); \label{for-pro12}
\end{eqnarray}
Combining the above inequalities, we directly have
\begin{eqnarray}
\phi_{0}(X, H)&\leq& \frac{1}{1-\mu}\phi_0 (S,H)\nonumber\\
&\leq&  \frac{\alpha}{1-\mu}\phi_0 (S,H_{opt})\nonumber\\
&\leq& \frac{\alpha (1+\mu)}{1-\mu}\phi_0 (X,H_{opt}).
\end{eqnarray}
So $H$ is an $\frac{\alpha (1+\mu)}{1-\mu}$-approximation of $X$. 

\section{Proof of Claim \ref{pro-sample}}
\label{sec-pro-sample}

Let the number of sampled elements be $h$. Since each sampled element falls in $V$ with probability $\tau$, by taking the union bound, we know that the sample contains at least one element from $V$ with probability $1-(1-\tau)^h$. Therefore, if we want $1-(1-\tau)^h\geq 1-\eta$, $h$ should be at least $\frac{\log 1/\eta}{\log 1/(1-\tau)}\leq\frac{1}{\tau}\log\frac{1}{\eta}$.

\section{Proof of Theorem \ref{the-biapprox}}
\label{sec-proof-the}

We assume that (1) in Lemma \ref{lem-select2} never happens, and prove that $\lambda_j(E)=k$ with constant probability when $j=\Theta(k)$. The idea actually has been used by~\cite{aggarwal2009adaptive} for obtaining a bi-criteria approximation for $k$-means clustering. Define a random variable $x_j$: $x_j=1$ if $\lambda_j(E)=\lambda_{j-1}(E)$, or $0$ if $\lambda_j(E)\geq\lambda_{j-1}(E)+1$, for $j=1, 2, \cdots$. So $\mathbb{E}[x_j]\leq\eta$ and
\begin{eqnarray}
\sum_{1\leq s\leq j}(1-x_s)\leq\lambda_j(E). \label{forr-azuma2}
\end{eqnarray}
Also, let $J_j=\sum_{1\leq s\leq j}(x_s-\eta)$ and $J_0=0$. Then, $\{J_0, J_1, J_2, \cdots\}$ is a super-martingale by the following definition.


\begin{definition}
\label{def-martingale}
A sequence of real valued random variables $J_0, J_1, \cdots, J_t$ is called a super-martingale if for every $j>1$, $E[J_j\mid J_0, \cdots, J_{j-1}]\leq J_{j-1}$.
\end{definition}

In addition, we know that $J_{j+1}-J_j\leq 1$ for each $j\geq 0$. 
Through {\em Azuma-Hoeffding inequality}~\cite{alon2004probabilistic}, we have 
\begin{eqnarray}
Pr(J_t\geq J_0+\delta)\leq e^{-\frac{\delta^2}{2t}}\label{forr-azuma1}
\end{eqnarray}
for any $t\in\mathbb{Z}^+$ and $\delta>0$. Let $t=\frac{k+\sqrt{k}}{1-\eta}$ and $\delta=\sqrt{k}$, (\ref{forr-azuma1}) implies that
\begin{eqnarray}
Pr(\sum_{1\leq s\leq t}(1-x_s)\geq t(1-\eta)-\delta)&\geq& 1-e^{-\frac{\delta^2}{2t}}\nonumber\\
\Longrightarrow Pr(\sum_{1\leq s\leq t}(1-x_s)\geq k)&\geq& 1-e^{-\frac{k}{2(k+\sqrt{k})/(1-\eta)}}\nonumber\\
\Longrightarrow Pr(\sum_{1\leq s\leq t}(1-x_s)\geq k)&\geq& 1-e^{-\frac{1-\eta}{4}}. \label{forr-azuma}
\end{eqnarray}
Combining (\ref{forr-azuma2}) and (\ref{forr-azuma}), we know that $\lambda_t(E)=k$ with probability at least $1-e^{-\frac{1-\eta}{4}}$. Moreover, $\lambda_t(E)=k$ directly implies that $E$ is a $2$-approximate solution by Claim~\ref{cla-e2}. Together with Lemma~\ref{lem-select1}, we have Theorem \ref{the-biapprox}.

\section{Proof of Claim~\ref{cla-core}}
\label{sec-proof-cla-core}
We just need to prove the first inequality since the other one can be obtained by the same manner. 
Because each $B_j\subseteq Ball(c_j, r_X)$ and each vertex $p$ is moved by a distance at most $\mu r_{opt}$ based on (\ref{for-map}), we know that $f(B_j)\subseteq Ball(c_j, r_X+\mu r_{opt})$, i.e., $r'_E\leq r_X+\mu r_{opt}$. 

Let $p_0$ be the vertex realizing $r_X=\phi_0(X, H)$, that is, there exists some $1\leq j_0\leq k$ such that $||c_{j_0}-p_0||=r_X$.  The triangle inequality and (\ref{for-map}) together imply $||c_{j_0}-f(p_0)||\geq r_X-\mu r_{opt}$. Hence $r'_E\geq r_X-\mu r_{opt}$.

Overall, we have $|r'_E-r_X|\leq \mu r_{opt}$.

\section{Algorithm~\ref{alg-bi} on Real Datasets}
\label{sec-exp-alg1}
\begin{table}[htbp]
	\centering
	\caption{Algorithm~\ref{alg-bi} on real dataset}
	\begin{tabular}{|c|c|c|c|c|c|c|c|c|c|c|c|}
		\hline
		\multicolumn{2}{|c|}{ $|E|/k$} & 3   & 6   & 9   & 12  & 15    & 18  & 21  & 24  & 27  & 30   \\
		\hline
		\multirow{2}[4]{*}{MNIST} & approx. ratio & 1.267 & 1.241 & 1.211 & 1.208 & 1.178  & 1.154 & 1.138 & 1.132 & 1.117 & 1.117 \\
		\cline{2-12}          & running time(s) & 7.110 & 7.154 & 7.199 & 7.086 & 7.055  & 7.166 & 7.263 & 7.148 & 7.192 & 7.180 \\
		\hline
		\multirow{2}[4]{*}{CALTECH256} & approx. ratio & 1.056 & 1.015 & 0.999 & 0.975 & 0.947  & 0.903 & 0.878 & 0.873 & 0.861 & 0.852  \\
		\cline{2-12}          & running time(s) & 2.747 & 2.743 & 2.878 & 2.802 & 2.758  & 3.726 & 2.727 & 2.674 & 2.704 & 2.653 \\
		\hline
		\multirow{2}[4]{*}{CIFAR10} & approx. ratio & 1.268 & 1.236 & 1.165 & 1.157 & 1.147  & 1.096 & 1.096 & 1.066 & 1.062 & 1.070 \\
		\cline{2-12}          & running time(s) & 32.946 & 32.910 & 32.049 & 32.145 & 32.178 & 32.514 & 32.886 & 32.590 & 34.364 & 32.625\\
		\hline
	\end{tabular}%
	\label{}%
\end{table}%

Note that it is possible to obtain approximation ratio lower than $1$, since we output more than $k$ cluster centers.

\end{document}